\documentclass[english,11pt,a4paper]{article}
\usepackage[T1]{fontenc}
\usepackage[latin9]{inputenc}
\usepackage{geometry}
\usepackage{babel}
\usepackage{hyperref}
\usepackage{enumitem}

\usepackage{amsmath, amssymb, amstext, mathtools}
\usepackage{amsthm}

\usepackage{tikz}
\usepackage{pgfplots}
\pgfplotsset{compat=1.18}
\usepackage{textpos}

\definecolor[named]{urlblue}{cmyk}{1,0.58,0,0.21}

\hypersetup{
    breaklinks=true,
    colorlinks=true,
    citecolor=purple!70!blue!60!black,
    linkcolor=purple!70!blue!90!black,
    urlcolor=urlblue,
    pdflang={en},
    pdftitle={FPT Isomorphism Test for F-Free Tournaments},
    pdfauthor={Daniel Neuen}
}

\usetikzlibrary{backgrounds,patterns,arrows,decorations.pathreplacing,decorations.pathmorphing,calc}

\tikzstyle{vertex}=[draw,circle,fill=white,minimum size=7pt,inner sep=0pt]

\newtheorem{theorem}{Theorem}[section]
\newtheorem{lemma}[theorem]{Lemma}

\newtheorem{observation}[theorem]{Observation}

\newtheorem{claim}[theorem]{Claim}
\newtheorem{fact}[theorem]{Fact}
\theoremstyle{definition}

\theoremstyle{remark}

\newenvironment{claimproof}{\begin{proof}}{\end{proof}}

\newcommand{\CC}{{\mathcal C}}

\newcommand{\CE}{{\mathcal E}}

\newcommand{\CM}{{\mathcal M}}

\newcommand{\CP}{{\mathcal P}}
\newcommand{\CQ}{{\mathcal Q}}

\newcommand{\NN}{{\mathbb N}}

\newcommand{\RR}{{\mathbb R}}

\newcommand{\C}{{\mathbb C}}

\newcommand{\WL}[2]{\chi^{#1,#2}}
\newcommand{\WLit}[3]{\chi_{(#2)}^{#1,#3}}

\newcommand{\trans}{{\mathsf{T}}}

\newcommand{\Alg}{{\mathbb A}}

\DeclareMathOperator{\im}{im}
\DeclareMathOperator{\tr}{tr}

\DeclareMathOperator{\Iso}{Iso}
\DeclareMathOperator{\Aut}{Aut}
\DeclareMathOperator{\Sym}{Sym}

\newcommand{\orcid}[1]{\href{https://orcid.org/#1}{\includegraphics[height=1.8ex]{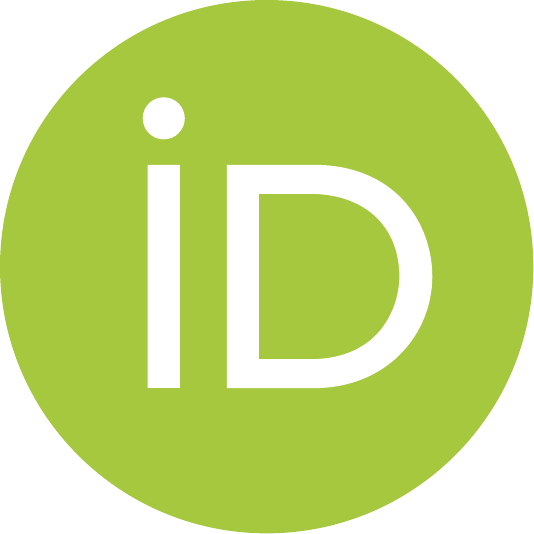}}}

\title{FPT Isomorphism Test for $F$-Free Tournaments}
\author{Daniel Neuen \orcid{0000-0002-4940-0318}\\
TU Dresden\\
Dresden, Germany}

\date{}

\begin{document}

\maketitle

\begin{abstract}
 We show that isomorphism of $F$-free tournaments can be solved in FPT time $f(k) \cdot n^{O(1)}$, where $k$ denotes the size of $F$, and $n$ denotes the size of the input tournaments.
 Our result extends a previous FPT isomorphism test for tournaments of bounded twin-width [Grohe, Neuen 2024], as well as XP isomorphism tests parameterized by the VC dimension or the chromatic number [Ra\ss mann, Schweitzer 2026].
 It also implies that every non-trivial hereditary class of tournaments admits a polynomial-time isomorphism test.

 Our algorithm builds on a combination of spectral, geometric, combinatorial and group-theoretic tools.
\end{abstract}

\section{Introduction}

The tournament isomorphism problem (TI) was recognized as a particularly interesting special case of the graph isomorphism problem (GI) early on.
Already in 1983, Babai and Luks~\cite{BabaiL83} proved that TI can be solved in quasipolynomial time $n^{O(\log n)}$;
it took 33 more years for Babai \cite{Babai16} to prove that the general GI can also be solved in quasipolynomial time.
An important fact that makes TI more accessible than GI is that every tournament has a solvable automorphism group.
This follows from the fact that the automorphism group of every tournament has odd order and the famous Feit-Thompson Theorem~\cite{FeitT63} stating that all groups of odd order are solvable.
However, even Babai's powerful machinery did not help us to improve the upper bound for TI.
But TI is not only special from a group-theoretic perspective.
Another remarkable result, due to Schweitzer~\cite{Schweitzer17}, states that TI reduces (under randomized polynomial-time reductions) to the problem of deciding whether a tournament has a non-trivial automorphism; the so-called \emph{rigidity problem}.
It is an open question whether the same holds for general graphs.

While there is an extensive literature on GI restricted to classes of graphs (see \cite{GroheN21,Neuen26} for recent surveys), restrictions of TI only started to receive more attention recently.
Ponomarenko \cite{Ponomarenko92} proved that TI is in polynomial time for tournaments whose automorphism group contains a regular cyclic subgroup,
and more recently Arvind, Ponomarenko, and Ryabov~\cite{ArvindPR25} proved that TI is in polynomial time for edge-colored tournaments where at least one edge color induces a (strongly) connected spanning subgraph of bounded degree (even fixed-parameter tractable when parameterized by the out-degree).
While both of these results are very interesting from a technical perspective, they consider fairly specialized classes of tournaments that would hardly be called natural from a graph-theoretic point of view.
Looking at natural graph parameters, Grohe and the present author~\cite{GroheN26} proved that TI can be solved in time $k^{O(\log k)}n^{O(1)}$ for tournaments of twin-width at most $k$.
Also, Ra{\ss}mann and Schweitzer \cite{RassmanS26} proved that TI for tournaments of VC dimension at most $d$ can be solved in time $n^{O(d \log d)}$, and TI for $k$-colorable tournaments can be solved in time $n^{O(k)}$.
Here, a tournament is $k$-colorable if the vertices can be colored with $k$ colors in such a way that every color class induces a transitive (i.e., acyclic) tournament.
We note that the VC dimension of a tournament is bounded in terms of the twin-width, i.e., classes of bounded twin-width also have bounded VC dimension.
As such, Ra{\ss}mann and Schweitzer \cite{RassmanS26} cover more general classes of tournaments, but they only obtain XP algorithms whereas for tournaments of twin-width $k$, the algorithm from \cite{GroheN26} is an FPT algorithm.

All three of these parameters $p$ (twin-width, VC dimension, chromatic number) have in common that they are \emph{hereditary}, i.e., for an induced subtournament $S$ of a tournament $T$, we have that $p(S) \leq p(T)$.
As the main result of this work, we generalize those results and obtain an FPT isomorphism test for all non-trivial hereditary classes of tournaments.
More precisely, for tournaments $F$ and $T$, we say that $T$ is \emph{$F$-free} if $T$ has no induced subtournament isomorphic to $F$.

\begin{theorem}
 \label{thm:main}
 The isomorphism problem for $F$-free tournaments can be solved in time $2^{2^{O(k^2)}} \cdot n^{O(1)}$ where $k \coloneqq |V(F)|$.
\end{theorem}

Note that for every non-trivial hereditary class $\CC$ of tournaments, there is some $F$ such that every graph in $\CC$ is $F$-free.
Hence, Theorem~\ref{thm:main} is a vast generalization of the polynomial-time isomorphism tests for classes of bounded twin-width, VC dimension or chromatic number.
Additionally, for the parameters VC dimension and chromatic number, our result gives the first FPT isomorphism test for tournaments.

Moreover, Theorem~\ref{thm:main} again highlights the special role of tournaments.
For general (undirected) graphs, there is a wide collection of non-trivial hereditary classes (e.g., bipartite graphs, chordal graphs, $2$-degenerate graphs, etc.) for which the isomorphism problem is as hard as on general graphs \cite{BoothC79} (see also, e.g., \cite{KratschS17,Schweitzer17a}).
In fact, if GI cannot be solved in polynomial time, then there are even subgraph-closed (and hence, hereditary) classes of graphs on which GI has intermediate complexity \cite{OtachiS13}.
In contrast, for tournaments, our result shows that every non-trivial hereditary class admits a polynomial-time isomorphism test.

The proof of Theorem~\ref{thm:main} relies on a combination of spectral, geometric, combinatorial and group-theoretic methods.
The starting point for the algorithm is the FPT isomorphism test for tournaments of bounded twin-width from \cite{GroheN26}.
The basic strategy of this algorithm is to proceed in roughly two phases (after some initial step that restricts to so-called $2$-WL-homogeneous tournaments; we ignore this step in the following overview).
First, the algorithm uses combinatorial methods to identify an isomorphism-invariant subset of the edges that defines a connected spanning graph of ``near-bounded degree''.
Giving those edges a special color, the second phase then solves isomorphism using group-theoretic methods which extend the algorithm from \cite{ArvindPR25}.

The algorithm for Theorem~\ref{thm:main} follows the same strategy, and in fact is almost identical to the algorithm from \cite{GroheN26}.
However, the key difference lies in the analysis, particularly in proving the existence of an isomorphism-invariant subset of the edges that defines a connected spanning graph of ``near-bounded degree''.
More concretely, to identify the desired subset of edges, both algorithms rely on a standard combinatorial graph isomorphism heuristic, the $2$-dimensional Weisfeiler-Leman algorithm (see, e.g., \cite{Kiefer20,Kiefer20a} for background), which colors pairs of vertices in an isomorphism-invariant manner.
Now, the key step is to prove the existence of a subset of $2$-WL colors $C$ such that taking all edges with a color from $C$ defines a subgraph of ``near-bounded degree''.
While this step is fairly straightforward for graphs of bounded twin-width, it turns out to be significantly more challenging for $F$-free tournaments.

In a nutshell, to prove the existence of the desired set $C$, we proceed as follows.
First, for an $F$-free tournament $T$, we consider the matrix
\[M \coloneqq (A - A^\trans)(A^\trans - A)\]
where $A$ is the adjacency matrix of $T$, and $A^\trans$ denotes the transpose.
This matrix is symmetric and positive semidefinite, i.e., all eigenvalues are non-negative.
We first show that the eigenspace of the largest eigenvalue has a dimension $2 \leq d \leq K(F)$ where $K(F)$ is some parameter that only depends on $F$.
Here, we also make use of arguments that are inspired by the study of quasi-random tournaments \cite{ChungG91,KalyanasundaramS13}.

Now, consider the projector $P \in \RR^{n \times n}$ onto the eigenspace of the largest eigenvalue of $M$.
Since the matrix $P$ has small rank, we can use arguments from \cite[Theorem 2]{EvdokimovP97} to build a graph of bounded degree.
For each vertex $v \in V(T)$, consider the point
\[p_v \coloneqq Pe_v\]
where $e_v \in \RR^{n}$ denotes the standard basis vector for position $v$, i.e., $(e_v)_v = 1$ and $(e_v)_w = 0$ for all $w \neq v$.
Crucially, it is not difficult to show that the Euclidean distance $\|p_v - p_w\|_2$ between two such points is preserved by every isomorphism.
Now, let
\[\delta \coloneqq \min_{v \neq w} \|p_v - p_w\|_2\]
denote the minimum distance between two such points.
Let us assume that $\delta > 0$ (for $\delta = 0$, we use a simple inductive argument which exploits that $d \geq 2$).
Now, to obtain a bounded-degree subgraph, we connect $v$ to all vertices $w$ such that $\|p_v - p_w\|_2 = \delta$.
The number of such neighbors is bounded, since every point $p_v$ lies in the eigenspace of the largest eigenvalue, which has bounded dimension $d$.
Indeed, since all points have pairwise distance at least $\delta$, a simple geometric packing argument shows that, for every $v \in V(T)$, there is only a bounded number of vertices $w \in V(T)$ such that $\|p_v - p_w\|_2 = \delta$.
This way, we obtain a bounded-degree subgraph $D$.
Let us stress that, to avoid any numerical issues, our algorithm does not compute any of the vectors $p_v$.
Instead, similar to \cite{EvdokimovP97}, we build on known arguments (see, e.g., \cite{Furer95,Furer10,RattanS23}) to show that the value $\|p_v - p_w\|_2$ is ``encoded'' in the $2$-WL colors.

Unfortunately, the graph $D$ may not yet be connected (which we require for the second phase to work).
To fix this, we iteratively add further edges to $D$ as follows.
We repeat the above procedure, but define $\delta$ as the minimum distance between points associated with vertices from different components of the current graph $D$.
Now, the geometric packing argument yields that, for every $v \in V(T)$, there is only a bounded number of components that $v$ has a neighbor in.
This results in a graph of ``near-bounded degree'' which is sufficient for our purposes.

\section{Preliminaries}

\subsection{Basics}

We write $\NN = \{0,1,2,3,\dots\}$ for the set of natural numbers and $[n] \coloneqq \{1,\dots,n\}$ for $n \in \NN$.
Let $V$ be a finite set and let $\CQ$ be a partition of $V$.
We write $\sim_\CQ$ for the unique equivalence relation such that the partition into equivalence classes equals $\CQ$.
For $X \subseteq V$, we write $\CQ \cap X \coloneqq \{Q \cap X \mid Q \in \CQ, Q \cap X \neq \emptyset\}$ for the partition induced on $X$.
For a coloring $\chi\colon V \to C$, where $C$ is some set of ``colors'', we write $\CP_\chi$ for the partition of $V$ into color classes.
For a second coloring $\chi'\colon V \to C$, we write $\chi' \preceq \chi$ if $\CP_{\chi'}$ is finer than $\CP_\chi$, i.e., for every $P' \in \CP_{\chi'}$ there is some $P \in \CP_{\chi}$ such that $P' \subseteq P$.
We say that $\chi$ is \emph{equivalent} to $\chi'$, denoted $\chi \equiv \chi'$, if $\CP_{\chi} = \CP_{\chi'}$.

Graphs in this paper are usually directed.
We denote the vertex set of a graph $G$ by $V(G)$ and the edge relation by $E(G)$.
Graphs are always loop-free, and there are no parallel edges.
For $X\subseteq V(G)$, we write $G[X]$ to denote the subgraph of $G$ induced on $X$.
A \emph{tournament} is a directed graph $T$ such that, for every distinct $v,w \in V(T)$, exactly one of the pairs $(v,w)$ and $(w,v)$ is an edge.
For a second tournament $F$, we say that $T$ is \emph{$F$-free} if there is no $X \subseteq V(T)$ such that $T[X]$ is isomorphic to $F$.

Let $G_1,G_2$ be two directed graphs.
An \emph{isomorphism} from $G_1$ to $G_2$ is a bijection $\varphi\colon V(G_1) \to V(G_2)$ such that $(v,w) \in E(G_1)$ if and only if $(\varphi(v),\varphi(w)) \in E(G_2)$ for all $v,w \in V(G_1)$.
We write $\varphi\colon G_1 \cong G_2$ to denote that $\varphi$ is an isomorphism from $G_1$ to $G_2$.
Also, $\Iso(G_1,G_2)$ denotes the set of all isomorphisms from $G_1$ to $G_2$.
The graphs $G_1$ and $G_2$ are \emph{isomorphic} if $\Iso(G_1,G_2) \neq \emptyset$.
The \emph{automorphism group} of $G_1$ is $\Aut(G_1) \coloneqq \Iso(G_1,G_1)$.

An \emph{arc coloring} of a (directed) graph $G$ is a mapping $\lambda\colon E(G) \to C$ for some set $C$ of ``colors''.
An \emph{arc-colored graph} is a triple $G=(V,E,\lambda)$, where $(V,E)$ is a graph and $\lambda$ is an arc coloring of $(V,E)$.
Isomorphisms between arc-colored graphs are required to preserve the coloring.

\subsection{Linear Algebra}

For our spectral arguments, we require some background in linear algebra.
We refer the reader to \cite{HornJ13} for further background.

We write $I_n$ for the $n \times n$ identity matrix over $\RR$, and $J_n$ for the $n \times n$ all-ones matrix.
For a matrix $A \in \RR^{n \times m}$, we write $A^\trans$ for the transpose of $A$.
For vectors $x,y \in \RR^n$ we write $\langle x,y \rangle \coloneqq x^\trans y$ and $\|x\|_2 = \sqrt{\langle x,x \rangle} = \left(\sum_{i=1}^n x_i^2\right)^{1/2}$, where $x_i$ denotes the $i$-th entry of $x$.
The Cauchy-Schwarz inequality states that
\begin{equation}
 \label{eq:cauchy-schwarz}
 |\langle x,y \rangle| \leq \|x\|_2 \cdot \|y\|_2
\end{equation}
for all $x,y \in \RR^{n}$.
For a square matrix $A \in \RR^{n \times n}$ we have
\[\|A\|_2 \coloneqq \max_{x \in \RR^n, \|x\|_2 = 1} \|Ax\|_2 = \max_{0 \neq x \in \RR^n} \frac{\|Ax\|_2}{\|x\|_2}.\]

Let $M \in \RR^{n \times n}$ be a symmetric matrix, i.e., $M = M^\trans$.
Then it is well-known that $M$ has $n$ real eigenvalues $\lambda_1 \geq \lambda_2 \geq \dots \geq \lambda_n$.
For an eigenvalue $\lambda$ of $M$, let us write
\[E_\lambda \coloneqq \{v \in \RR^{n} \mid M \cdot v = \lambda \cdot v\}\]
for the eigenspace of $\lambda$.
The dimension of $E_\lambda$ is equal to the multiplicity of the eigenvalue $\lambda$.
We write $\lambda_{\max}(M)$ for the largest eigenvalue of $M$.
For each eigenspace of the matrix $M$, there is a projector $P$ to the eigenspace which can be written as a polynomial in $M$.
More precisely, we shall use the following fact.

\begin{fact}
 \label{fact:projector}
 Let $M \in \RR^{n \times n}$ be a symmetric matrix, and let $\lambda_1,\dots,\lambda_r$ denote the distinct eigenvalues such that $\lambda_1 = \lambda_{\max}(M)$.
 Let
 \[P \coloneqq \prod_{j = 2}^{r} \frac{1}{\lambda_1 - \lambda_j}(M - \lambda_jI_n)\]
 be the \emph{projector onto $E_{\lambda_1}$}, which satisfies $\im(P) \coloneqq \{Pv \mid v \in \RR^{n}\} = E_{\lambda_1}$.
\end{fact}

The matrix $M$ is \emph{positive semidefinite} if $x^\trans M x \geq 0$ for every $x \in \RR^{n}$.
This is equivalent to all eigenvalues of $M$ being non-negative.

The \emph{trace} of $M$ is the sum of diagonal entries, i.e., $\tr(M) \coloneqq \sum_{i=1}^n M_{i,i}$.
We have
\begin{equation}
 \label{eq:trace-eigenvalue}
 \tr(M) = \sum_{i = 1}^{n} \lambda_i.
\end{equation}

For every square matrix $A \in \RR^{n \times n}$, we have that $AA^\trans$ is symmetric and positive semidefinite.
Also, the norm of $A$ can be expressed via the largest eigenvalue of $AA^\trans$.

\begin{fact}
 \label{fact:eigenvalue-vs-norm}
 Let $A \in \RR^{n \times n}$ be a matrix.
 Then
 \[\|A\|_2^2 = \lambda_{\max}(AA^\trans).\]
\end{fact}

Finally, we shall need the following fact on skew-symmetric matrices.

\begin{fact}
 \label{fact:even-multiplicity}
 Let $S \in \RR^{n \times n}$ be a skew-symmetric matrix, i.e., $S^\trans = -S$.
 Then every non-zero eigenvalue of $SS^\trans$ has even multiplicity.
\end{fact}

\subsection{Weisfeiler-Leman Algorithm}

Next, we describe the $2$-dimensional Weisfeiler-Leman algorithm, originally introduced by Weisfeiler and Leman \cite{WeisfeilerL68} (see also \cite{Weisfeiler76}), which is used as a key subroutine in the main algorithm.

Let $G$ be a directed graph.
For $i \geq 0$, we describe the coloring $\WLit{2}{i}{G}$ of $(V(G))^2$ computed in the $i$-th iteration of $2$-WL.
For $i = 0$, each pair is colored with the isomorphism type of the underlying ordered induced subgraph.
So if $H$ is another directed graph and $v_1,v_2 \in V(G)$, $w_1,w_2 \in V(H)$
then $\WLit{2}{0}{G}(v_1,v_2) = \WLit{2}{0}{H}(w_1,w_2)$ if and only if, for all $i,j \in \{1,2\}$, it holds that $v_i = v_j \Leftrightarrow w_i = w_j$ and $(v_i,v_j) \in E(G) \Leftrightarrow (w_i,w_j) \in E(H)$.
If $G$ and $H$ are arc-colored, then the colors are also taken into account.

Now let $i \geq 0$.
For $v_1,v_2 \in V(G)$ we define
\[\WLit{2}{i+1}{G}(v_1,v_2) \coloneqq \Big(\WLit{2}{i}{G}(v_1,v_2), \CM_{i}(v_1,v_2)\Big)\]
where
\[\CM_i(v_1,v_2) \coloneqq \Big\{\!\Big\{ \big(\WLit{2}{i}{G}(x,v_2),\WLit{2}{i}{G}(v_1,x)\big) ~\Big|~ x \in V(G) \Big\}\!\Big\}\]
(here $\{\!\{\dots\}\!\}$ denotes a multiset).

Clearly, $\WLit{2}{i+1}{G} \preceq \WLit{2}{i}{G}$ for all $i \geq 0$.
So there is a unique minimal $i_\infty \geq 0$ such that $\WLit{2}{i_\infty+1}{G} \equiv \WLit{2}{i_\infty}{G}$ and we write $\WL{2}{G} \coloneqq \WLit{2}{i_\infty+1}{G}$ to denote the corresponding coloring.

The $2$-dimensional Weisfeiler-Leman ($2$-WL) algorithm takes as input a (possibly colored) graph $G$ and outputs (a coloring that is equivalent to) $\WL{2}{G}$.
This can be done in time $O(n^3 \log n)$ \cite{ImmermanL90}.

Let $H$ be a second graph.
The $2$-WL algorithm \emph{distinguishes} $G$ and $H$ if there is a color $c$ such that
\[\Big|\Big\{ \bar v \in (V(G))^2 ~\Big|~ \WL{2}{G}(\bar v) = c \Big\}\Big| \neq \Big|\Big\{ \bar w \in (V(H))^2 ~\Big|~ \WL{2}{H}(\bar w) = c \Big\}\Big|.\]
We write $G \simeq_2 H$ to denote that $2$-WL does not distinguish between $G$ and $H$.

A key focus of our main algorithm is $2$-WL-homogeneous tournaments.
A directed graph $G$ is \emph{$2$-WL-homogeneous} if for all $v,w\in V(G)$ it holds that $\WL{2}{G}(v,v) = \WL{2}{G}(w,w)$.
We require the following basic fact on $2$-WL-homogeneous directed graphs.

\begin{fact}[{see, e.g., \cite[Lemma 3.5]{GroheN26}}]
 \label{fact:wl-weakly-to-strongly-cc}
 Let $G$ be a $2$-WL-homogeneous directed graph, and let $C$ be a set of colors in the range of $\WL{2}{G}$.
 Let $G[C]$ denote the directed graph with vertex set $V(G[C]) \coloneqq V(G)$ and edge set
 \[E(G[C]) \coloneqq \{(v,w) \in E(G) \mid \WL{2}{G}(v,w) \in C\}.\]
 Then the weakly connected components of $G[C]$ equal the strongly connected components of $G[C]$.
\end{fact}

Let $G$ be a directed graph and $\CQ$ a partition of $V(G)$.
We say that $\CQ$ is \emph{$\WL{2}{G}$-definable} if there is some set $C$ of colors such that
\[v \sim_\CQ w \quad\iff\quad \WL{2}{G}(v,w) \in C\]
for all $v,w \in V(G)$.
For a $\WL{2}{G}$-definable partition $\CQ$, a \emph{cross-$\CQ$ color} is a color $c$ in the image of $\WL{2}{G}$ that is not contained in $C$.
We shall require the following well-known fact (see, e.g., \cite{ChenP19}).

\begin{fact}
 \label{fact:wl-partition}
 Let $T$ be a $2$-WL-homogeneous, arc-colored tournament and let $\CQ,\CE$ be two $\WL{2}{T}$-definable partitions of $V(T)$.
 Also let $X \in \CE$.
 Define $S = (T[X],\lambda_S)$ to be the subtournament induced by $X$ with arc-coloring $\lambda_S(x,y) \coloneqq \WL{2}{T}(x,y)$ for all $(x,y) \in E(S) = E(T[X])$.

 Then $\CQ \cap X$ is $\WL{2}{S}$-definable and
 \[\WL{2}{T}(v,w) = \WL{2}{T}(x,y) \quad\implies\quad \WL{2}{S}(v,w) = \WL{2}{S}(x,y)\]
 for all $v,w,x,y \in X$.
\end{fact}

We note that a suitable version also holds for arbitrary digraphs, if the arc-coloring is replaced by a coloring of all pairs.

We shall also use an alternative description of the Weisfeiler-Leman algorithm based on so-called \emph{coherent algebras} (dating back to Higman \cite{Higman87}).
Let $G$ be a directed graph and let $A_G \in \{0,1\}^{V(G) \times V(G)}$ denote its adjacency matrix.
The \emph{coherent algebra} of $G$ is the algebra $\Alg(G) \subseteq \C^{V(G) \times V(G)}$ obtained by closing the set $\{A_G,A_G^\trans,I_n,J_n\}$ under scalar multiplication, matrix addition, matrix multiplication, and Hadamard multiplication (i.e., for matrices $B,C \in \C^{n \times n}$, the Hadamard product is defined as $(B \circ C)_{i,j} \coloneqq B_{i,j} \cdot C_{i,j}$).

If $G = (V,E,\lambda)$ is an arc-colored graph, we additionally add, for every color $c$ in the image of $\lambda$, the matrices $A_{G,c}$ and $A_{G,c}^\trans$ where $A_{G,c} \in \{0,1\}^{V(G) \times V(G)}$ is defined via $(A_{G,c})_{v,w} = 1$ if and only if $(v,w) \in E$ and $\lambda(v,w) = c$.

The following fact connects the $2$-WL algorithm to the coherent algebra of $G$ (see, e.g., \cite{ChenP19}).

\begin{fact}
 \label{fact:wl-coherent-algebra}
 Let $G$ be an arc-colored directed graph, and let $v,w,x,y \in V(G)$.
 Then
 \[\WL{2}{G}(v,w) = \WL{2}{G}(x,y) \quad\iff\quad M_{v,w} = M_{x,y} ~\forall M \in \Alg(G).\]
\end{fact}

\subsection{Group Theory}

Next, we describe the required tools from the group-theoretic machinery.
For basic background on permutation groups and their computational aspects, we also refer to \cite{DixonM96,Rotman99,Seress03}.

Let $\Gamma$ be a group and let $\gamma,\delta \in \Gamma$.
The \emph{commutator} of $\gamma$ and $\delta$ is $[\gamma,\delta] \coloneqq \gamma^{-1}\delta^{-1}\gamma\delta$.
The \emph{commutator subgroup $[\Gamma,\Gamma]$} of $\Gamma$ is the unique subgroup of $\Gamma$ generated by all commutators $[\gamma,\delta]$ for $\gamma,\delta \in \Gamma$.
Note that $[\Gamma,\Gamma]$ is a normal subgroup of $\Gamma$.
The \emph{derived series of $\Gamma$} is the sequence of subgroups $\Gamma^{(0)} \trianglerighteq \Gamma^{(1)} \trianglerighteq \Gamma^{(2)} \trianglerighteq \dots$ where $\Gamma^{(0)} \coloneqq \Gamma$ and $\Gamma^{(i+1)} \coloneqq [\Gamma^{(i)},\Gamma^{(i)}]$ for all $i \geq 0$.
A group $\Gamma$ is \emph{solvable} if there is some $i \geq 0$ such that $\Gamma^{(i)}$ is the trivial group (i.e., it only contains the identity element).

By the Feit-Thompson Theorem \cite{FeitT63} every group of odd order is solvable.
Also, for every tournament $T$, the automorphism group $\Aut(T)$ has odd order.
Indeed, $\Aut(T)$ cannot contain an involution (a permutation of order $2$), since every involution swaps some pair $v,w$ of distinct vertices and exactly one of $(v,w), (w,v)$ is an edge of $T$.
Together, we obtain the following:

\begin{theorem}
 \label{thm:aut-solvable}
 Let $T$ be a tournament.
 Then $\Aut(T)$ is solvable.
\end{theorem}

This theorem is useful, since several computational problems on permutation groups can be solved in polynomial time on solvable groups, while for general permutation groups, no polynomial-time algorithms are currently known.

Let $G_1$ and $G_2$ be two (arc-colored) directed graphs.
Also let $\Gamma \leq \Sym(V(G_1))$ be a permutation group and let $\theta\colon V(G_1) \to V(G_2)$ be a bijection.
Observe that $\Gamma\theta \coloneqq \{\gamma\theta \mid \gamma \in \Gamma\}$ is a set of bijections from $V(G_1)$ to $V(G_2)$.
We define
\[\Iso_{\Gamma\theta}(G_1,G_2) \coloneqq \Iso(G_1,G_2) \cap \Gamma\theta = \{\varphi \in \Gamma\theta \mid \varphi\colon G_1 \cong G_2\}\]
and $\Aut_\Gamma(G_1) \coloneqq \Iso_\Gamma(G_1,G_1)$.
Note that $\Aut_\Gamma(G_1) \leq \Gamma$ and, if $\Iso_{\Gamma\theta}(G_1,G_2) \neq \emptyset$, then $\Iso_{\Gamma\theta}(G_1,G_2) = \Aut_\Gamma(G_1)\varphi$ where $\varphi \in \Iso_{\Gamma\theta}(G_1,G_2)$ is an arbitrary isomorphism from $G_1$ to $G_2$.

\begin{theorem}[{\cite[Corollary 3.6]{BabaiL83}}]
 \label{thm:gi-solvable-group}
 Let $G_1 = (V_1,E_1,\lambda_1)$ and $G_2 = (V_2,E_2,\lambda_2)$ be two arc-colored directed graphs.
 Also let $\Gamma \leq \Sym(V_1)$ be a solvable group and $\theta\colon V_1 \to V_2$ a bijection.
 Then $\Iso_{\Gamma\theta}(G_1,G_2)$ can be computed in polynomial time.
\end{theorem}

We stress that $\Gamma$ is given by a set of generators (typically of size polynomial in $n = |V(G_1)|$).
Also note that $\Iso_{\Gamma\theta}(G_1,G_2)$ may be of size exponential in the number of vertices of $G_1$ and $G_2$.
However, if $\Iso_{\Gamma\theta}(G_1,G_2) \neq \emptyset$, $\Iso_{\Gamma\theta}(G_1,G_2) = \Aut_\Gamma(G_1)\varphi$ where $\varphi \in \Iso_{\Gamma\theta}(G_1,G_2)$ is an arbitrary isomorphism from $G_1$ to $G_2$.
Hence, the set $\Iso_{\Gamma\theta}(G_1,G_2)$ can be represented by a generating set for $\Aut_\Gamma(G_1)$ of size polynomial in $|V(G_1)|$ and a single element $\varphi \in \Iso_{\Gamma\theta}(G_1,G_2)$.
Let us stress at this point that all isomorphism sets considered in this work are represented in this way.

\section{Main Structural Lemma}

In the following, we design an FPT isomorphism test for $F$-free tournaments.
The algorithm essentially consists of three parts.
First, we use well-established group-theoretic methods to reduce to the case where both input tournaments are $2$-WL-homogeneous (the exact same step is also used, e.g., in \cite{GroheN26}).
For $2$-WL-homogeneous tournaments, we proceed in two phases.
First, we use spectral and geometric arguments to show that the coloring produced by the $2$-WL algorithm admits a subset $C$ of colors that define a subgraph (i.e., we pick exactly the edges with a color from $C$) that has a ``near-bounded-degree'' structure.
Then, in the second phase, we use group-theoretic methods to solve isomorphism by exploiting the discovered ``near-bounded-degree'' structure.
In fact, the second phase of the algorithm is essentially identical to \cite{GroheN26} where the same ``near-bounded-degree'' structure is exploited to design an FPT isomorphism test for tournaments of twin-width at most $k$.
Hence, the main technical contribution of this paper lies in the first phase which discovers the ``near-bounded-degree'' structure.

In this section, we cover the first phase of the algorithm and prove the main structural lemma (Lemma~\ref{lem:main-structural}).
To formulate the lemma, we first need to define some parameters for the forbidden tournament $F$.
Let $F$ be a tournament and let $k \coloneqq |V(F)|$.
We define
\[m \coloneqq \binom{k}{2} \quad\text{and}\quad K(F) \coloneqq 4^{m} \quad\text{and}\quad D(F) \coloneqq 3^{K(F)} - 1.\]
Observe that $D(F) = 2^{2^{O(k^2)}}$.

\begin{lemma}
 \label{lem:main-structural}
 Let $F$ be a tournament.
 Also let $T$ be an $F$-free, arc-colored tournament that is $2$-WL-homogeneous.
 Let $\CQ$ be a $\WL{2}{T}$-definable partition of $V(T)$ such that $|\CQ| \geq 2$.

 Then there is a cross-$\CQ$ color $c \in \{\WL{2}{T}(x,y) \mid (x,y) \in E(T)\}$ such that
 \begin{equation}
  \label{eq:main-structural}
  |\{Q \in \CQ \mid \exists w \in Q\colon \WL{2}{T}(v,w) = c\}| \leq D(F)
 \end{equation}
 for every $v \in V(T)$.
\end{lemma}

The remainder of this section is devoted to the proof of Lemma~\ref{lem:main-structural}.
In a first step, we prove that the eigenspace of the largest eigenvalue of a matrix associated with $T$ has dimension $2 \leq d \leq K(F)$.
Afterwards, we use a geometric packing argument in the eigenspace (going back to \cite{EvdokimovP97}) to bound the number of blocks in \eqref{eq:main-structural}.

\subsection{Bounding the Multiplicity of the Largest Eigenvalue}

As the first step towards proving Lemma~\ref{lem:main-structural}, we bound the dimension of the eigenspace for the largest eigenvalue of the matrix $M_T$ defined below.

Let $T$ be a tournament.
We write $A = A_T$ for the adjacency matrix of $T$, i.e., $A_{v,w} = 1$ if and only if $(v,w) \in E(T)$.
Also, we write
\[S_T \coloneqq A_T - A_T^\trans\]
for the \emph{skew adjacency matrix}.
Observe that $S_T^\trans = -S_T$.
Finally, we define
\[M_T \coloneqq S_T \cdot S_T^\trans = -S_T^2\]
which is symmetric and positive semidefinite.
In particular, all its eigenvalues are real and non-negative, and we obtain the following basic observation.

\begin{observation}
 \label{obs:bound-eigenvalue-sum}
 Let $T$ be an $n$-vertex tournament and $M \coloneqq M_T$.
 Let $\lambda_1 \geq \lambda_2 \geq \dots \geq \lambda_n \geq 0$ denote the eigenvalues of $M$.
 Then $M_{v,v} = n-1$ for all $v \in V(T)$.
 In particular,
 \[n(n-1) = \tr(M) = \sum_{i = 1}^{n} \lambda_i.\]
\end{observation}

\begin{proof}
 Let us write $S \coloneqq S_T$.
 Let $v \in V(T)$.
 Then
 \begin{align*}
  M_{v,v} = \sum_{w \in V(T)} S_{v,w}S^\trans_{w,v} = \sum_{w \in V(T)} S_{v,w}S_{v,w} = \sum_{v \neq w \in V(T)} 1 = n-1
 \end{align*}
 since $S_{v,w} \in \{-1,1\}$ for all $v \neq w \in V(T)$, and $S_{v,v} = 0$ for all $v \in V(T)$.
 Hence, $n(n-1) = \tr(M) = \sum_{i = 1}^{n} \lambda_i$ where the second equality follows from \eqref{eq:trace-eigenvalue}.
\end{proof}

Let $F$ and $T$ be tournaments.
We write $\hom(F,T)$ for the number of homomorphisms from $F$ to $T$, i.e., the number of mappings $\varphi\colon V(F) \to V(T)$ such that $(\varphi(v),\varphi(w)) \in E(T)$ for all $(v,w) \in E(F)$.
Note that, since $F$ is a tournament and $T$ does not have self-loops, every homomorphism from $F$ to $T$ is injective.
In particular, $\hom(F,T) = 0$ if and only if $T$ is $F$-free.
We start by proving the following bound that is inspired by the study of quasi-random tournaments \cite{ChungG91,KalyanasundaramS13}.

\begin{lemma}
 \label{lem:hom-estimate}
 Let $F$ and $T$ be tournaments.
 Also define $k \coloneqq |V(F)|$, $m \coloneqq \binom{k}{2} = |E(F)|$ and $n \coloneqq |V(T)|$.
 Then
 \begin{equation}
  |\hom(F,T) - 2^{-m}n^k| \leq 2n^{k-1}\|B_T\|_2
 \end{equation}
 where $B_T \coloneqq A_T - \frac{1}{2}J_n$.
\end{lemma}

\begin{proof}
 Let us write $A \coloneqq A_T$ and $B \coloneqq B_T$.
 Also, we may assume that $V(F) = [k]$.
 Finally, we fix an arbitrary enumeration $(a_1,b_1),\dots,(a_m,b_m)$ of the edges of $F$.
 We have that
 \[\hom(F,T) = \sum_{v_1,\dots,v_k \in V(T)} \prod_{i \in [m]} A_{v_{a_i},v_{b_i}}.\]
 Hence,
 \begin{equation}
  \label{eq:hom-estimate-1}
  \hom(F,T) - 2^{-m}n^k = \sum_{v_1,\dots,v_k \in V(T)}\left( - 2^{-m} + \prod_{i \in [m]} A_{v_{a_i},v_{b_i}} \right).
 \end{equation}

 \begin{claim}
  Let $m \geq 1$ and let $z_1,\dots,z_m \in \RR$. Then
  \[- 2^{-m} + \prod_{i \in [m]} z_i = \sum_{i=1}^{m} 2^{-m+i} \cdot \left(z_i - \frac{1}{2}\right) \cdot \prod_{j < i}z_j.\]
 \end{claim}
 \begin{claimproof}
  We prove the statement by induction on $m$.
  The base case $m = 1$ is trivial.
  For $m > 1$, we have that
  \begin{align*}
   \sum_{i=1}^{m} 2^{-m+i} \cdot \left(z_i - \frac{1}{2}\right) \cdot \prod_{j < i}z_j
   &= \left(z_m - \frac{1}{2}\right) \cdot \prod_{j < m}z_j ~~+~~ \frac{1}{2} \cdot \sum_{i=1}^{m-1} 2^{-(m-1)+i} \cdot \left(z_i - \frac{1}{2}\right) \cdot \prod_{j < i}z_j\\
   &= \left(z_m - \frac{1}{2}\right) \cdot \prod_{j < m}z_j ~~+~~ \frac{1}{2}\left(- 2^{-(m-1)} + \prod_{i \in [m-1]} z_i\right)\\
   &= - 2^{-m} ~~+~~ \frac{1}{2}\prod_{i \in [m-1]} z_i ~~-~~ \frac{1}{2}\prod_{j < m}z_j ~~+~~ z_m\prod_{j < m}z_j\\
   &= - 2^{-m} ~~+~~ \prod_{i \in [m]} z_i.\qedhere
  \end{align*}
 \end{claimproof}
 Plugging this into \eqref{eq:hom-estimate-1}, we obtain
 \begin{equation}
  \label{eq:hom-estimate-2}
  \hom(F,T) - 2^{-m}n^k ~~=~~ \sum_{i=1}^{m} 2^{-m+i} \sum_{v_1,\dots,v_k \in V(T)} \left(A_{v_{a_i},v_{b_i}} - \frac{1}{2}\right)\prod_{j < i} A_{v_{a_j},v_{b_j}}
 \end{equation}
 Now, let us consider some $i \in [m]$.
 We claim that
 \begin{equation}
  \label{eq:hom-estimate-3}
  \left|\sum_{v_1,\dots,v_k \in V(T)} \left(A_{v_{a_i},v_{b_i}} - \frac{1}{2}\right)\prod_{j < i} A_{v_{a_j},v_{b_j}}\right| ~~\leq~~ n^{k-1} \cdot \|B\|_2.
 \end{equation}
 First observe that this implies the lemma, since plugging \eqref{eq:hom-estimate-3} into \eqref{eq:hom-estimate-2} gives that
 \[|\hom(F,T) - 2^{-m}n^k| \leq \sum_{i=1}^{m} 2^{-m+i} \cdot n^{k-1} \cdot \|B\|_2 = n^{k-1} \cdot \|B\|_2 \cdot \sum_{i=1}^{m} 2^{-m+i} \leq 2 \cdot n^{k-1} \cdot \|B\|_2.\]
 So it remains to prove \eqref{eq:hom-estimate-3}.
 Let us fix some $i \in [m]$ and let $p \coloneqq a_i$, $q \coloneqq b_i$.
 Now, we can write
 \begin{equation}
  \label{eq:hom-estimate-4}
  \sum_{v_1,\dots,v_k \in V(T)} \left(A_{v_{a_i},v_{b_i}} - \frac{1}{2}\right)\prod_{j < i} A_{v_{a_j},v_{b_j}} ~~=~~ \sum_{\substack{v_s \in V(T)\\s=1,\dots,k\\p \neq s \neq q}} \sum_{v_p,v_q \in V(T)} B_{v_p,v_q} \prod_{j < i} A_{v_{a_j},v_{b_j}}.
 \end{equation}
 Now, let us fix vertices $v_s$ for every $s \in \{1,\dots,k\}$ with $p \neq s \neq q$.
 Then there are functions $f,g \colon V(T) \to \{0,1\}$ such that, for every $v_p,v_q \in V(T)$, we have
 \[\prod_{j < i} A_{v_{a_j},v_{b_j}} = f(v_p) \cdot g(v_q).\]
 Hence, interpreting $f,g$ as vectors in $\RR^{V(T)}$, we obtain that
 \begin{equation}
  \label{eq:hom-estimate-5}
  \sum_{v_p,v_q \in V(T)} B_{v_p,v_q} \prod_{j < i} A_{v_{a_j},v_{b_j}} ~~=~~ \sum_{v_p,v_q \in V(T)} f(v_p)B_{v_p,v_q}g(v_q) ~~=~~ f^\trans B g.
 \end{equation}
 By the Cauchy-Schwarz inequality, we conclude that
 \[|f^\trans B g| = |\langle f,Bg \rangle| \leq \|f\|_2 \cdot \|Bg\|_2 \leq \|f\|_2 \cdot \|B\|_2 \cdot \|g\|_2.\]
 Since $f,g$ are 0/1-vectors, it follows that $\|f\|_2 \leq \sqrt{n}$ and $\|g\|_2 \leq \sqrt{n}$.
 Hence, $|f^\trans B g| \leq n \cdot \|B\|_2$.
 Combining this with \eqref{eq:hom-estimate-4} and \eqref{eq:hom-estimate-5}, we obtain that
 \[\left|\sum_{v_1,\dots,v_k \in V(T)} \left(A_{v_{a_i},v_{b_i}} - \frac{1}{2}\right)\prod_{j < i} A_{v_{a_j},v_{b_j}}\right| ~~\leq~~ n^{k-2} \cdot n \cdot \|B\|_2 = n^{k-1} \cdot \|B\|_2\]
 proving \eqref{eq:hom-estimate-3}.
\end{proof}

Now, we use the last lemma to bound the dimension of the eigenspace for the largest eigenvalue of $M_T$.

\begin{lemma}
 \label{lem:eigenspace-dimension}
 Let $F$ be a fixed tournament.
 Also let $T$ be an $F$-free tournament with $n \geq 2$ vertices.
 Let $\lambda \coloneqq \lambda_{\max}(M_T)$ denote the largest eigenvalue of $M_T$, and let $d$ denote its multiplicity (i.e., the dimension of its eigenspace).
 Then
 \[2 \leq d \leq K(F).\]
\end{lemma}

\begin{proof}
 Let us write $A \coloneqq A_T$, $B \coloneqq B_T = A_T - \frac{1}{2}J_n$, $S = S_T$ and $M \coloneqq M_T$.
 Also let $k \coloneqq |V(F)|$ and $m \coloneqq \binom{k}{2} = |E(F)|$.
 Since $T$ is $F$-free, we conclude that $\hom(F,T) = 0$.
 Hence, by Lemma~\ref{lem:hom-estimate}, we obtain that
 \[2^{-m}n^k \leq 2 \cdot n^{k-1} \cdot \|B\|_2\]
 which implies that
 \begin{equation}
  \|B\|_2 \geq 2^{-(m+1)} \cdot n.
 \end{equation}
 Now, $J_n = A + A^\trans + I_n$, since $T$ is a tournament.
 Hence,
 \[B = A - \frac{1}{2}J_n = A - \frac{1}{2}(A + A^\trans + I_n) = \frac{1}{2}(A - A^\trans - I_n) = \frac{1}{2}(S - I_n).\]
 It follows that
 \[BB^\trans = \frac{1}{4}(S - I_n)(S^\trans - I_n) = \frac{1}{4}(SS^\trans - S - S^\trans + I_n) = \frac{1}{4}(M + I_n)\]
 where the last equality holds since $S = -S^\trans$.
 Using Fact~\ref{fact:eigenvalue-vs-norm}, we conclude that
 \[\frac{1}{4}\left(\lambda_{\max}(M) + 1\right) = \lambda_{\max}\left(\frac{1}{4}(M + I_n)\right) = \lambda_{\max}(BB^\trans) = \|B\|_2^2.\]
 It follows that
 \[\lambda_{\max}(M) = 4\|B\|_2^2 - 1 \geq 4 \cdot (2^{-(m+1)} \cdot n)^2 - 1 = \frac{n^2}{4^m} - 1 = \frac{n^2}{K(F)} - 1\]
 By Observation~\ref{obs:bound-eigenvalue-sum}, we obtain that
 \[d \cdot \lambda_{\max}(M) \leq n(n-1).\]
 If $K(F) \leq n$, then we conclude that
 \[d \leq n(n-1) \frac{K(F)}{n^2 - K(F)} = K(F) \frac{n^2 - n}{n^2 - K(F)} \leq K(F).\]
 Otherwise, $K(F) > n$ and clearly $d \leq n < K(F)$.

 Finally, recall that $M = SS^\trans$ and $S$ is a skew-symmetric matrix.
 Hence, since $\lambda_{\max}(M) \neq 0$, we conclude that $d \geq 2$ by Fact~\ref{fact:even-multiplicity}.
\end{proof}

\subsection{Weisfeiler-Leman and Spectral Invariants}

In Lemma~\ref{lem:eigenspace-dimension} we have established that the eigenspace corresponding to the largest eigenvalue of $M_T$ has small dimension.
To exploit this fact for the proof of Lemma~\ref{lem:main-structural}, we still need to establish that certain properties associated with the eigenspace are ``encoded'' in the $2$-WL coloring.

Let us remark at this point that, as an alternative approach, one could also work directly with the required spectral invariants.
However, at this point, one starts to encounter certain numerical issues since all eigenvalues, eigenvectors and projectors need to be computed to a sufficiently large precision.
To avoid such issues, we follow a strategy similar to \cite{Furer95} and instead rely on the purely combinatorial $2$-WL algorithm which is known to capture many spectral invariants (see, e.g., \cite{Furer95,Furer10,RattanS23}).
In particular, we obtain the following lemma.

\begin{lemma}
 \label{lem:wl-projection-distance}
 Let $T$ be a $2$-WL-homogeneous, arc-colored tournament, and let $M \coloneqq M_T$.
 Let $P$ be the projector onto $E_\lambda$ where $\lambda = \lambda_{\max}(M)$ is the largest eigenvalue of $M$.
 Also, for $x \in V(T)$, let
 \[p_x \coloneqq P \cdot e_x \in \RR^{n}.\]
 Then
 \[\WL{2}{T}(x,y) = \WL{2}{T}(v,w) \quad\implies\quad \|p_x - p_y\|_2 = \|p_v - p_w\|_2\]
 for every $x,y,v,w \in V(T)$.
\end{lemma}

\begin{proof}
 By definition (see Fact~\ref{fact:projector}), we have that $P \in \Alg(T)$.
 Let $x,y,v,w \in V(T)$ such that $\WL{2}{T}(x,y) = \WL{2}{T}(v,w)$.
 By definition, this implies that
 \begin{align*}
     &\Big\{\!\Big\{ \big(\WL{2}{T}(z,y),\WL{2}{T}(x,z)\big) ~\Big|~ z \in V(T) \Big\}\!\Big\}\\
  =~~&\Big\{\!\Big\{ \big(\WL{2}{T}(z,w),\WL{2}{T}(v,z)\big) ~\Big|~ z \in V(T) \Big\}\!\Big\}
 \end{align*}
 Hence, there is a bijection $f\colon V(T) \to V(T)$ such that
 \begin{equation*}
  \big(\WL{2}{T}(z,y),\WL{2}{T}(x,z)\big) ~~=~~ \big(\WL{2}{T}(f(z),w),\WL{2}{T}(v,f(z))\big)
 \end{equation*}
 for all $z \in V(T)$.
 By the properties of the $2$-WL algorithm, we also obtain that
 \begin{equation*}
  \big(\WL{2}{T}(z,y),\WL{2}{T}(z,x)\big) ~~=~~ \big(\WL{2}{T}(f(z),w),\WL{2}{T}(f(z),v)\big)
 \end{equation*}
 By Fact~\ref{fact:wl-coherent-algebra}, we conclude that
 \[\big(P_{z,y},P_{z,x}\big) ~~=~~ \big(P_{f(z),w},P_{f(z),v}\big)\]
 for all $z \in V(T)$.
 Now, we get that
 \[\|p_x - p_y\|_2^2 = \sum_{z \in V(T)} (P_{z,x} - P_{z,y})^2 = \sum_{z \in V(T)} (P_{f(z),v} - P_{f(z),w})^2 = \sum_{z \in V(T)} (P_{z,v} - P_{z,w})^2 = \|p_v - p_w\|_2^2.\]
 Taking the square root of the both sides gives the desired result.
\end{proof}

\subsection{Proof of Main Structural Lemma}

Now, we are ready to prove the main structural lemma.
We note that the argument in the case $\delta > 0$ is essentially identical to that in the proof of \cite[Theorem 2]{EvdokimovP97}.

\begin{proof}[Proof of Lemma~\ref{lem:main-structural}]
 We prove the statement by induction on $n$.
 The base case $n = 1$ is trivial.

 So suppose that $n \geq 2$.
 Let $M \coloneqq M_T$ and let $\lambda \coloneqq \lambda_{\max}(M)$.
 Also let $P$ be the projector onto $E_\lambda$.
 Also, for $x \in V(T)$, let
 \[p_x \coloneqq P \cdot e_x \in \RR^{n}.\]
 We define
 \begin{equation}
  \delta \coloneqq \min_{x \not\sim_\CQ y} \| p_x - p_y\|_2.
 \end{equation}
 Now, we distinguish two cases.
 \begin{description}
  \item[Case $\delta > 0$:]
   Pick $x,y \in V(T)$ such that $x \not\sim_\CQ y$, $(x,y) \in E(T)$ and $\delta = \| p_x - p_y\|_2$.
   We set $c \coloneqq \WL{2}{T}(x,y)$.
   Then $c$ is a cross-$\CQ$ color by definition.
   It remains to verify \eqref{eq:main-structural}.
   Let $v \in V(T)$ and suppose that
   \[\{Q \in \CQ \mid \exists w \in Q\colon \WL{2}{T}(v,w) = c\} = \{Q_1,\dots,Q_\ell\}.\]
   For every $i \in [\ell]$ fix an arbitrary $w_i \in Q_i$ such that $\WL{2}{T}(v,w_i) = c$.
   Let
   \[p_0 \coloneqq p_v \quad\text{and}\quad p_i \coloneqq p_{w_i} \text{ for every } i \in [\ell].\]
   Observe that
   \begin{equation}
    \label{eq:point-distance-1}
    \|p_0 - p_i\|_2 = \delta
   \end{equation}
   for every $i \in [\ell]$ by Lemma~\ref{lem:wl-projection-distance}.
   Also, by the definition of $\delta$, we conclude that
   \begin{equation}
    \label{eq:point-distance-2}
    \|p_i - p_j\|_2 \geq \delta
   \end{equation}
   for all $i \neq j \in [\ell]$, since $w_i \not\sim_\CQ w_j$.
   Since $p_0,\dots,p_\ell \in \RR^{n}$ span a space of dimension at most $K(F)$, we can now use a geometric packing argument to bound the number of points.

   For $i \in \{0,\dots,\ell\}$ let
   \[U_i \coloneqq \{v \in \im(P) \mid \|p_i - v\|_2 < \delta/2\}.\]
   Note that $U_i \cap U_j = \emptyset$ for all distinct $i,j \in \{0,\dots,\ell\}$ by \eqref{eq:point-distance-1} and \eqref{eq:point-distance-2}.
   Also, let
   \[U_0^{\text{large}} \coloneqq \{v \in \im(P) \mid \|p_0 - v\|_2 < 3\delta/2\}.\]
   Then $U_i \subseteq U_0^{\text{large}}$ for all $i \in \{0,\dots,\ell\}$ by \eqref{eq:point-distance-1}.
   Now, let $d$ denote the dimension of the vector space $\im(P)$, and let $\omega_d$ denote the volume of a $d$-dimensional ball with radius $1$.
   Then $U_i$ has volume $\omega_d \cdot (\delta/2)^d$ for every $i \in \{0,\dots,\ell\}$, and $U_0^{\text{large}}$ has volume $\omega_d \cdot (3\delta/2)^d$.
   Altogether, it follows that
   \[(\ell + 1) \cdot \omega_d \cdot (\delta/2)^d \leq \omega_d \cdot (3\delta/2)^d\]
   which implies that
   \[\ell + 1 \leq 3^d.\]
   Since $d \leq K(F)$ by Lemma~\ref{lem:eigenspace-dimension}, we conclude that $\ell \leq 3^{K(F)} - 1 = D(F)$.
  \item[Case $\delta = 0$:]
   For two vertices $v,w \in V(T)$ we define $v \sim_{\CE} w$ if $p_v = p_w$.
   Note that $\sim_{\CE}$ is an equivalence relation and let $\CE$ denote the partition of $V(T)$ into equivalence classes.
   Also observe that $\CE$ is $\WL{2}{T}$-definable by Lemma~\ref{lem:wl-projection-distance}.

   Since $\delta = 0$, there are vertices $x,y \in V(T)$ such that $x \not\sim_\CQ y$, but $x \sim_\CE y$.
   Let $X \in \CE$ denote the partition class such that $x,y \in X$.
   We first argue that $X \neq V(T)$.
   Indeed, if $X = V(T)$, then $p_v = p_w$ for all $v \in V(T)$, i.e., $\im(P)$ is a 1-dimensional space.
   However, this contradicts Lemma~\ref{lem:eigenspace-dimension} stating that the dimension of $\im(P)$ is at least $2$.

   Now define $S = (T[X],\lambda_S)$ to be the subtournament induced by $X$ with arc-coloring $\lambda_S(v,w) \coloneqq \WL{2}{T}(v,w)$ for all $(v,w) \in E(S) = E(T[X])$.
   Clearly, $S$ is also $F$-free since it is an induced subtournament of $T$.
   Also, $S$ is $2$-WL-homogeneous and $\CQ \cap X$ is $\WL{2}{S}$-definable by Fact~\ref{fact:wl-partition}.
   Finally, $|\CQ \cap X| \geq 2$ since $x \not\sim_\CQ y$.

   By the induction hypothesis, there is a cross-$(\CQ \cap X)$-color $c_S \in \{\WL{2}{S}(x,y) \mid (x,y) \in E(S)\}$ such that
   \[|\{Q \in \CQ \cap X \mid \exists w \in Q\colon \WL{2}{S}(v,w) = c_S\}| \leq D(F)\]
   for every $v \in V(S) = X$.
   Fix some $x,y \in X$ such that $\WL{2}{S}(x,y) = c_S$; we set $c \coloneqq \WL{2}{T}(x,y)$.

   For every $y' \in V(T)$ such that $\WL{2}{T}(x,y') = c$, we have that $y' \in X$, since $\CE$ is $\WL{2}{T}$-definable, and $\WL{2}{S}(x,y') = c_S$ by Fact~\ref{fact:wl-partition}.
   It follows that
   \begin{equation}
    \label{eq:main-structural-x}
    |\{Q \in \CQ \mid \exists w \in Q\colon \WL{2}{T}(x,w) = c\}| \leq D(F),
   \end{equation}
   i.e., \eqref{eq:main-structural} is satisfied for the vertex $x$.
   To complete the proof, we show the following claim.

   \begin{claim}
    \label{claim:main-structural-homogeneous}
    For every $v \in V(T)$ it holds that
    \[|\{Q \in \CQ \mid \exists w \in Q\colon \WL{2}{T}(x,w) = c\}| = |\{Q \in \CQ \mid \exists w \in Q\colon \WL{2}{T}(v,w) = c\}|.\]
   \end{claim}
   \begin{claimproof}
    Let $v \in V(T)$.
    First observe that $\WL{2}{T}(x,x) = \WL{2}{T}(v,v)$, since $T$ is $2$-WL-homogeneous.
    This implies that
    \[a \coloneq |\{w \in V(T) \mid \WL{2}{T}(x,w) = c\}| = |\{w \in V(T) \mid \WL{2}{T}(v,w) = c\}|,\]
    i.e., $x$ and $v$ have the same number $a$ of $c$-colored outgoing edges.

    Because $\CQ$ is $\WL{2}{T}$-definable, there is a set of colors $C$ such that
    \[w \sim_\CQ w' \quad\iff\quad \WL{2}{T}(w,w') \in C\]
    for all $w,w' \in V(T)$.
    Now, there is a number $b \in \NN$ such that, for every pair $(y,w) \in (V(T))^2$ such that $\WL{2}{T}(y,w) = c$, we have that
    \begin{align*}
     b &= |\{w' \in V(T) \mid \WL{2}{T}(y,w') = c, \WL{2}{T}(w,w') \in C\}|\\
       &= |\{w' \in V(T) \mid \WL{2}{T}(y,w') = c, w \sim_\CQ w'\}|.
    \end{align*}
    It follows that
    \[|\{Q \in \CQ \mid \exists w \in Q\colon \WL{2}{T}(v,w) = c\}| = \frac{a}{b}\]
    for all $v \in V(T)$, which proves the claim.
   \end{claimproof}
   Now, \eqref{eq:main-structural} is satisfied for all $v \in V(T)$ by Claim~\ref{claim:main-structural-homogeneous} and \eqref{eq:main-structural-x}.\qedhere
 \end{description}
\end{proof}

\section{Isomorphism Tests for Hereditary Classes}

In this section, we give an FPT isomorphism test for $F$-free tournaments based on Lemma~\ref{lem:main-structural}.
We stress that, building on Lemma~\ref{lem:main-structural}, the algorithm is essentially identical to the FPT algorithm from \cite{GroheN26}.
In particular, we rely on the following lemma from \cite{GroheN26}.
Let $T = (V,E,\lambda)$ be an arc-colored tournament, and let $\CQ$ be a partition.
For a set of colors $C$ in the image of $\lambda$, we say that \emph{$\CQ$ is $\lambda$-defined by $C$} if, for every $(v,w) \in E(T)$ we have that $v \sim_\CQ w$ if and only if $\lambda(v,w) \in C$.
Observe that $\CQ$ is uniquely determined by $C$, since $T$ is a tournament.
If $\CQ$ is $\lambda$-defined by $C$, then every other color in the image of $\lambda$ is called a \emph{cross-$\CQ$ color}.

\begin{lemma}[{\cite[Lemma 4.1]{GroheN26}}]
 \label{lem:lift-isomorphisms}
 There is an algorithm that, given
 \begin{enumerate}[label = (\Alph*)]
  \item\label{item:lift-isomorphisms-1} an integer $d \geq 1$;
  \item\label{item:lift-isomorphisms-2} two arc-colored tournaments $T_1=(V_1,E_1,\lambda_1)$ and $T_2=(V_2,E_2,\lambda_2)$;
  \item\label{item:lift-isomorphisms-3} a set of colors $C$ and for $j=1,2$ a partition $\CQ_j$ of $V_j$ that is $\lambda_j$-defined by $C$;
  \item\label{item:lift-isomorphisms-4} a color $c^*$ that is a cross-$\CQ_j$ color for both $j=1,2$ and
   \begin{itemize}
    \item for every $v \in V_j$ it holds that
     \[\big|\big\{Q \in \CQ_j ~\big|~ \exists w \in Q\colon (v,w) \in E_j \wedge \lambda_j(v,w) = c^*\big\}\big| \leq d;\]
    \item for
     \[F_j \coloneqq \big\{(Q,Q')\in\CQ_j^2 ~\big|~ Q\neq Q', \exists w \in Q,w' \in Q'\colon (w,w') \in E_j \wedge \lambda_j(w,w') = c^*\big\}\]
     the directed graph $G_j = (\CQ_j,F_j)$ is strongly connected;
   \end{itemize}
  \item\label{item:lift-isomorphisms-5} $\Iso\big(T_j[Q],T_{j'}[Q']\big)$ for every $j,j' \in \{1,2\}$ and every $Q \in \CQ_j$, $Q' \in \CQ_{j'}$,
 \end{enumerate}
 computes $\Iso(T_1,T_2)$ in time $d^{O(\log d)} \cdot n^{O(1)}$.
\end{lemma}

Now, we are ready to describe the FPT isomorphism test for $F$-free tournaments.
We first obtain an algorithm for $2$-WL-homogeneous $F$-free tournaments, and then generalize it to arbitrary $F$-free tournaments.

\begin{lemma}
 \label{lem:isomorphism-unicolored}
 There is an algorithm that, given a tournament $F$ and $2$-WL-homogeneous tournaments $T_1,T_2$, either concludes that $T_1$ is not $F$-free or computes $\Iso(T_1,T_2)$ in time $2^{2^{O(k^2)}} \cdot n^{O(1)}$, where $k \coloneqq |V(F)|$ and $n \coloneqq \max(|V(T_1)|,|V(T_2)|)$.
\end{lemma}

\begin{proof}
 We run $2$-WL on $T_1$ and $T_2$.
 If $2$-WL distinguishes between $T_1$ and $T_2$, then we return that $\Iso(T_1,T_2) = \emptyset$.
 So suppose that $T_1 \simeq_2 T_2$.

 For $j \in \{1,2\}$ we define an arc-coloring via $\lambda_j(v,w) \coloneqq \WL{2}{T_j}(v,w)$ for every $(v,w) \in E(T_j)$.
 We write $\widehat{T}_j = (V(T_j),E(T_j),\lambda_j)$ for the corresponding arc-colored version of $T_j$.
 Observe that $\Iso(T_1,T_2) = \Iso(\widehat{T}_1,\widehat{T}_2)$.

 For $j \in \{1,2\}$, we define $\CQ_{j,0} \coloneqq \{\{v\} \mid v \in V(T_j)\}$ to be the partition into singletons.
 Both partitions are trivially $\lambda_j$-defined by the empty set of colors.
 Now, suppose the partitions $\CQ_{1,i}$ and $\CQ_{2,i}$ for some $i \geq 0$ have already been defined, and $|\CQ_{1,i}| \geq 2$.
 Also, suppose both sets are $\lambda_j$-defined by the set $C_i$.
 By Lemma~\ref{lem:main-structural}, there is a cross-$\CQ_{1,i}$ color $c_{i+1}$ such that, for every $v \in V(T_1)$,
 \[|\{Q \in \CQ_{1,i} \mid \exists w \in Q \colon \WL{2}{T_1}(v,w) = c_{i+1}\}| \leq D(F).\]
 We fix an arbitrary such color $c_{i+1}$ (if no such color exists, we return that $T_1$ is not $F$-free).
 Finally, for $j \in \{1,2\}$, we define $\CQ_{j,i+1}$ as the partition into the (weakly) connected components of the graph with vertex set $V(T_j)$ and all edges $(v,w) \in E(T_j)$ such that $\lambda_j(v,w) \in C_i \cup \{c_{i+1}\}$.
 By the properties of the $2$-WL algorithm, $\CQ_{j,i}$ is strictly finer than $\CQ_{j,i+1}$ and $\CQ_{j,i+1}$ (for both $j \in \{1,2\}$) is $\lambda_j$-defined by some set of colors $C_{i+1} \supseteq C_i \cup \{c_{i+1}\}$.
 Observe that, by Fact~\ref{fact:wl-weakly-to-strongly-cc}, the weakly connected components of this graph equal its strongly connected components.
 Moreover, since $T_1 \simeq_2 T_2$, we obtain that
 \[|\{Q \in \CQ_{2,i} \mid \exists w \in Q \colon \WL{2}{T_2}(v,w) = c_{i+1}\}| \leq D(F).\]
 We stop this process for the minimal $\ell \geq 0$ such that $|\CQ_{1,\ell}| = 1$.

 Next, for every $i \in \{0,\dots,\ell\}$, every $j,j' \in \{1,2\}$, and every $Q \in \CQ_{j,i}$, $Q' \in \CQ_{j',i}$ we inductively compute the set $\Iso(\widehat{T}_j[Q],\widehat{T}_{j'}[Q'])$.
 For $i = 0$ this is trivial since $|Q| = |Q'| = 1$.

 So suppose $i \in \{0,\dots,\ell-1\}$, $j,j' \in \{1,2\}$, and $Q \in \CQ_{j,i+1}$, $Q' \in \CQ_{j',i+1}$.
 Let $d \coloneqq D(F)$.
 We set $\widetilde{T}_1 \coloneqq \widehat{T}_j[Q]$ and $\widetilde{T}_2 \coloneqq \widehat{T}_{j'}[Q']$.
 Also, we set
 \[\widetilde{\CQ}_1 \coloneqq \{\widetilde{Q} \in \CQ_{j,i} \mid \widetilde{Q} \subseteq Q\}\]
 and
 \[\widetilde{\CQ}_2 \coloneqq \{\widetilde{Q} \in \CQ_{j',i} \mid \widetilde{Q} \subseteq Q'\}.\]
 Note that $\widetilde{\CQ}_1$ is a partition of $V(\widetilde{T}_1)$ and $\widetilde{\CQ}_2$ is a partition of $V(\widetilde{T}_2)$.
 Finally, we set $c^* \coloneqq c_{i+1}$ and apply Lemma~\ref{lem:lift-isomorphisms} to compute the set
 \[\Iso(\widehat{T}_j[Q],\widehat{T}_{j'}[Q']) = \Iso(\widetilde{T}_1,\widetilde{T}_2)\]
 in time $d^{O(\log d)} \cdot n^{O(1)} = 2^{2^{O(k^2)}} \cdot n^{O(1)}$.
 Observe that all the isomorphism sets required in Lemma~\ref{lem:lift-isomorphisms}\ref{item:lift-isomorphisms-5} have already been computed in the previous iteration.

 Finally, observe that $\CQ_{1,\ell} = \{V(T_1)\}$ and $\CQ_{2,\ell} = \{V(T_2)\}$.
 So in iteration $i = \ell$ the algorithm computes
 \[\Iso(\widehat{T}_1,\widehat{T}_2) = \Iso(T_1,T_2).\]

 To bound the running time of the algorithm, one can observe that all steps can be performed in polynomial time except for the subroutine in Lemma~\ref{lem:lift-isomorphisms}.
 However, this subroutine is called at most $\ell \cdot (|V(T_1)| + |V(T_2)|)^2$ times, which is polynomial in the input size, and runs in time $2^{2^{O(k^2)}} \cdot n^{O(1)}$.
 Overall, this gives the desired running time.
\end{proof}

Finally, relying on well-established arguments, we remove the requirement for the input tournaments to be $2$-WL-homogeneous.
We reformulate our main algorithmic result (Theorem~\ref{thm:main}) in a slightly more precise manner.

\begin{theorem}
 \label{thm:isomorphism}
 There is an algorithm that, given tournaments $F$ and $T_1,T_2$, either concludes that $T_1$ is not $F$-free or computes $\Iso(T_1,T_2)$ in time $2^{2^{O(k^2)}} \cdot n^{O(1)}$, where $k \coloneqq |V(F)|$ and $n \coloneqq \max(|V(T_1)|,|V(T_2)|)$.
\end{theorem}

\begin{proof}
 We run $2$-WL on $T_1$ and $T_2$.
 If $2$-WL distinguishes between $T_1$ and $T_2$, then we return that $\Iso(T_1,T_2) = \emptyset$.
 So suppose that $T_1 \simeq_2 T_2$.

 Let $C_V \coloneqq \{\WL{2}{T_1}(v,v) \mid v \in V(T_1)\}$ denote the set of \emph{vertex colors} of the coloring $\WL{2}{T_1}$.
 For $c \in C_V$ and $j \in \{1,2\}$ we define
 \[X_{j,c} \coloneqq \{v \in V(T_j) \mid \WL{2}{T_j}(v,v) = c\}.\]
 For every $c \in C_V$ we use the algorithm from Lemma~\ref{lem:isomorphism-unicolored} to compute the set
 \[\Gamma_c\theta_c \coloneqq \Iso(T_1[X_{1,c}],T_2[X_{2,c}]).\]
 If there is some $c \in C_V$ for which the algorithm from Lemma~\ref{lem:isomorphism-unicolored} concludes that $T_1[X_{1,c}]$ is not $F$-free, then $T_1$ is also not $F$-free.

 Also, if there is a color $c \in C_V$ for which $\Iso(T_1[X_{1,c}],T_2[X_{2,c}]) = \emptyset$, then $\Iso(T_1,T_2) = \emptyset$.
 Otherwise, let
 \[\Gamma \coloneqq \bigtimes_{c \in C_V} \Gamma_c \leq \Sym(V(T_1))\]
 and $\theta\colon V(T_1) \to V(T_2)$ be the bijection defined via $\theta(v) \coloneqq \theta_c(v)$ for the unique $c \in C_V$ such that $v \in X_{1,c}$.
 Then
 \[\Iso(T_1,T_2) \subseteq \Gamma\theta.\]
 Observe that the group $\Gamma_c$ is solvable for every $c \in C_V$ by Theorem~\ref{thm:aut-solvable}.
 Hence, $\Gamma$ is also solvable.
 So
 \[\Iso(T_1,T_2) = \Iso_{\Gamma\theta}(T_1,T_2)\]
 can be computed in polynomial time using Theorem~\ref{thm:gi-solvable-group}.
\end{proof}

\section*{AI Statement}

The results in this work were obtained through a discussion with ChatGPT-6 Astra.
In response to an initial prompt, ChatGPT-6 Astra provided an XP algorithm (parameterized by the size of $F$) building on Lemma~\ref{lem:eigenspace-dimension}, but was unable to obtain an FPT algorithm.
The FPT algorithm was obtained in a subsequent discussion where the main additional ideas and high-level proof structure were provided by the author.
In particular, this included the use of the ``near-bounded-degree'' structure and partition hierarchy (inspired by \cite{GroheN26}) as well as the use of the $2$-WL algorithm to avoid numerical computations (and accompanying bit-complexity issues) which were part of the initially proposed XP algorithm (inspired by \cite{Furer95}).

Finally, I stress that the write-up does not contain any AI-generated text, and is fully written by the author (based on the results of the discussion with ChatGPT-6 Astra).

\bibliographystyle{plainurl}
\bibliography{literature}

\end{document}